\documentclass{article}

\usepackage[preprint]{neurips_2020}

\usepackage[utf8]{inputenc} % allow utf-8 input
\usepackage[T1]{fontenc}    % use 8-bit T1 fonts
\usepackage{hyperref}       % hyperlinks
\usepackage{url}            % simple URL typesetting
\usepackage{booktabs}       % professional-quality tables
\usepackage{amsfonts}       % blackboard math symbols
\usepackage{nicefrac}       % compact symbols for 1/2, etc.
\usepackage{microtype}      % microtypography

\usepackage{amsmath}
\usepackage{amsthm}
\usepackage{amssymb}
\usepackage{amsfonts}
\usepackage{tikz}
\usepackage{pgfplots}
\pgfplotsset{compat=1.7}
\usepgfplotslibrary{fillbetween}
\usepackage{fontawesome}

\newtheorem{problem*}{Problem}
\newtheorem{theorem}{Theorem}

\newtheorem{definition}{Definition}
\newtheorem{proposition}{Proposition}

\newtheorem*{example*}{Example}
\usepackage{euscript}
\usepackage[bb=boondox]{mathalfa}
\usepackage[flushleft]{threeparttable}
\usepackage[linesnumbered,ruled]{algorithm2e}
\usepackage{wrapfig}
\usepackage{graphicx}
\usepackage{lipsum}
\usepackage{todonotes}
\usepackage{color}
\usepackage{setspace}
\usepackage{verbatim}
\usepackage{fancyvrb}
\usepackage{adjustbox}
\usepackage{multirow}

\newcommand{\cM}{\mathcal{M}}

\newcommand{\minimize}[1]{\underset{{#1}}{\text{minimize}}}

\newcommand{\bracketit}[1]{\left[{#1}\right]}
\newcommand{\braceit}[1]{\left({#1}\right)}
\newcommand{\st}{\text{subject to}}
\newcommand{\setcard}[1]{\scalebox{0.5}{$|#1|$}}
\usepackage{mathtools}
\newcommand\norm[1]{\left\lVert#1\right\rVert}
\usepackage{esvect}
\DeclareMathOperator{\Tr}{Tr}
\newcommand{\se}[1]{\text{s}(#1)}
\newcommand{\re}[1]{\text{r}(#1)}

\pgfmathdeclarefunction{laplace}{2}{%
  \pgfmathparse{-1/(2*#2)*exp(-abs(x-#1)/#2)}%
}
\definecolor{maincolor}{HTML}{032F99}
\definecolor{blue}{RGB}{31,64,122}
\definecolor{red}{HTML}{e05a87} 
\makeatletter
\newcommand{\oset}[3][0ex]{%
  \mathrel{\mathop{#3}\limits^{
    \vbox to#1{\kern-2\ex@
    \hbox{$\scriptstyle#2$}\vss}}}}
\makeatother
\newcommand{\optimal}[1]{\oset{\scalebox{.5}{$\star$}}{#1}\!}
\newcommand{\loss}{{\scriptstyle\Delta}c}
\usepackage{letltxmacro}
\LetLtxMacro\orgvdots\vdots
\LetLtxMacro\orgddots\ddots

\makeatletter
\DeclareRobustCommand\vdots{%
  \mathpalette\@vdots{}%
}
\newcommand*{\@vdots}[2]{%
  \sbox0{$#1\cdotp\cdotp\cdotp\m@th$}%
  \sbox2{$#1.\m@th$}%
  \vbox{%
    \dimen@=\wd0 %
    \advance\dimen@ -3\ht2 %
    \kern.5\dimen@
    \dimen@=\wd2 %
    \advance\dimen@ -\ht2 %
    \dimen2=\wd0 %
    \advance\dimen2 -\dimen@
    \vbox to \dimen2{%
      \offinterlineskip
      \copy2 \vfill\copy2 \vfill\copy2 %
    }%
  }%
}
\DeclareRobustCommand\ddots{%
  \mathinner{%
    \mathpalette\@ddots{}%
    \mkern\thinmuskip
  }%
}
\newcommand*{\@ddots}[2]{%
  \sbox0{$#1\cdotp\cdotp\cdotp\m@th$}%
  \sbox2{$#1.\m@th$}%
  \vbox{%
    \dimen@=\wd0 %
    \advance\dimen@ -3\ht2 %
    \kern.5\dimen@
    \dimen@=\wd2 %
    \advance\dimen@ -\ht2 %
    \dimen2=\wd0 %
    \advance\dimen2 -\dimen@
    \vbox to \dimen2{%
      \offinterlineskip
      \hbox{$#1\mathpunct{.}\m@th$}%
      \vfill
      \hbox{$#1\mathpunct{\kern\wd2}\mathpunct{.}\m@th$}%
      \vfill
      \hbox{$#1\mathpunct{\kern\wd2}\mathpunct{\kern\wd2}\mathpunct{.}\m@th$}%
    }%
  }%
}
\makeatother

\pgfplotsset{
    mark repeat*/.style={
        scatter,
        scatter src=x,
        scatter/@pre marker code/.code={
            \pgfmathtruncatemacro\usemark{
                or(mod(\coordindex,#1)==0, (\coordindex==(\numcoords-1))
            }
            \ifnum\usemark=0
                \pgfplotsset{mark=none}
            \fi
        },
        scatter/@post marker code/.code={}
    }
}

\makeatletter
\newcommand*{\rom}[1]{\expandafter\@slowromancap\romannumeral #1@}
\makeatother

\title{Differentially Private Convex Optimization\\ with Feasibility Guarantees}

\author{%
    Vladimir Dvorkin\\
    Technical University of Denmark\\
    Lyngby, Denmark \\
    \texttt{vladvo@elektro.dtu.dk} \\
    \And
    Ferdinando Fioretto \\
    Syracuse University \\
    Syracuse, NY, USA \\
    \texttt{ffiorett@syr.edu} \\
    \And
    Pascal Van Hentenryck \\
    Georgia Institute of Technology \\
    Atlanta, GA, USA \\
    \texttt{pvh@isye.gatech.edu} \\
    \AND
    Jalal Kazempour \\
    Technical University of Denmark\\
    Lyngby, Denmark \\
    \texttt{seykaz@elektro.dtu.dk} \\
    \And
    Pierre Pinson\\
    Technical University of Denmark\\
    Lyngby, Denmark \\
    \texttt{ppin@elektro.dtu.dk}
}

\begin{document}
\begingroup
\allowdisplaybreaks

\maketitle

\begin{abstract}
This paper develops a novel differentially private framework to solve convex optimization problems with sensitive optimization data and complex physical or operational constraints. Unlike standard noise-additive algorithms, that act primarily on the problem data, objective or solution, and disregard the problem constraints, this framework requires the optimization variables to be a function of the noise and exploits a chance-constrained problem reformulation with formal feasibility guarantees. The noise is calibrated to provide differential privacy for identity and linear queries on the optimization solution. For many applications, including resource allocation problems, the proposed framework provides a trade-off between the expected optimality loss and the variance of optimization results.
\end{abstract}

\section{Introduction}
Differential privacy \citep{Dwork:06} is a rigorous definition of privacy that quantifies and bounds the risk of disclosing sensitive attributes of datasets used in computations. Differentially private algorithms ensure privacy by introducing a calibrated noise to the inputs, outputs, or objectives of computations. It has been successfully applied to a variety of contexts, including histogram queries \citep{li2010optimizing}, census surveys \citep{abowd2018us,Fioretto:cp-19}, linear regression \citep{chaudhuri:2011} and deep learning \citep{abadi2016deep} to name but a few examples. However, their applications to \emph{constrained} optimization problems remains limited, because it is generally hard to certify the feasibility of differentially private optimization solution.  

This paper considers a parametric constrained optimization problems of the form
\begin{align}\label{base_problem}
        \minimize{z}\quad& c(z)\quad
        \st\quad \mathcal{Z}\triangleq\{z\;|\;Az\leqslant b,\; Gz=d\},
\end{align}
with variables $z \in \mathbb{R}_+^{n}$, convex cost function $c:\mathbb{R}^{n}\mapsto\mathbb{R}$, and convex, compact and non-empty feasible space $\mathcal{Z}$ with parameters $A\in\mathbb{R}^{m\times n}$, $b\in\mathbb{R}^{m}$, $G\in\mathbb{R}^{\ell\times n}$, and $d\in\mathbb{R}^{\ell}$, with $m > 0$ and $\ell > 0$. The paper assumes elements $c, A, b$ and $G$ as public, non-sensitive information about the system design, whereas vector $d=\{d_{i}\}_{i=1}^{\ell}$ contains private, sensitive data of every user $i$, e.g. the customer loads in an electrical power system. In such applications, the feasible space $\mathcal{Z}$ encodes hard operational constraints or physical laws that must be satisfied. 

Releasing queries over the solutions of problem \eqref{base_problem} may leak information about the sensitive data $d$. For example, in energy network operations, releasing the nodal energy supplies using identity queries, or aggregated supply quantities using sum queries, exposes the allocation of energy demand $d$ across the network \citep{zhou2019differential}. Therefore, the goal of this work is to compute such solution $z$ that makes queries over $z$ differentially private, while also being feasible for problem constraints. 

While there exist various differential privacy algorithms to solve convex optimization problems, their application to constrained problems is limited because they do not generally guarantee that the privacy-preserving result necessarily satisfies the feasibility conditions. Algorithms based on input perturbation of the sensitive data $d$ modify the feasible space $\mathcal{Z}$ \citep{Dwork:06,fukuchi2017differentially},
thus returning an approximate solution to \eqref{base_problem} that may not satisfy the original constraints. The output perturbation mechanisms \citep{chaudhuri2009privacy,rubinstein2012learning}, that add noise to the optimization results, generally cannot be certified feasible for $\mathcal{Z}$; see, for example, the impossibility results of \citet{hsu2014privately}. The feasibility and near-optimality of the privacy-preserving results can be restored by leveraging the post-processing immunity of differential privacy. This, however, requires solving bilevel optimization problems \citep{fioretto:TPS20}.

This paper addresses these limitations and develops a new framework that provides both privacy and feasibility guarantees for constrained optimization problems. Instead of applying the noise to either the parameters or the results of the optimization, the framework solves a stochastic chance-constrained optimization problem whose solution is used to sample a solution to \eqref{base_problem}, which guarantees privacy and ensures feasibility with high probability. The approach requires a linear dependency between the optimization variables and the noise \citep{georghiou2019decision} and reveals a novel connection between differential privacy and stochastic chance-constrained optimization. 
%%%%%%%%%%%%%%%%%%%%%%%%%
\begin{wrapfigure}[12]{R}{0.28\textwidth}
\centering
\vspace{-26pt}
\hspace*{-6pt}
\includegraphics[width=0.28\textwidth]{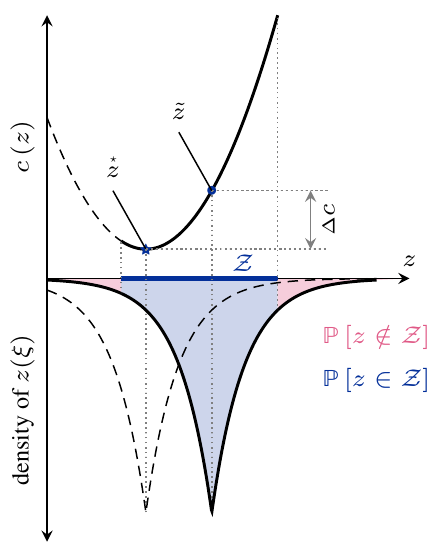}
\end{wrapfigure}
\ \ % keep this two to avoid the error of wrapfigure 

%%%%%%%%%%%%%%%%%%%%

The functioning of the framework is illustrated in the adjacent figure, showing the projections of solutions $z$ onto cost function $c$ and feasible space $\mathcal{Z}$. 
Consider the optimal solution $\optimal{z}\;$ returned by problem \eqref{base_problem}. The output perturbation of $\optimal{z}\;$ results in solutions whose density $\optimal{z}(\xi)$ (dashed line on the bottom of the figure) is prone to lie outside the feasible space $\mathcal{Z}$. 
By restricting the optimization variables to be a linear function of the noise, e.g., $z(\xi)=\tilde{z}+f(\xi)$ where $\tilde{z}$ is the expected value of the solution with respect to the random variable distribution and $f(\xi)$ is the linear functional recourse, the stochastic problem optimizes $\tilde{z}$ and $f(\xi)$ providing a new probability density of solution $z$ (solid line). This new density renders any realization of the noise $\xi$ feasible within a \emph{prescribed} probability $\mathbb{P}\bracketit{\tilde{z}+f(\xi)\in\mathcal{Z}}$, specified by the curator of problem \eqref{base_problem}. If the functional recourse $f(\xi)$ is made independent from the sensitive data $d$, the framework enjoys both privacy and feasibility guarantees. However, it introduces a trade-off between privacy and the optimality loss with expected value  $\loss=\mathbb{E}[c(z(\xi)) - c(\optimal{z}\;)]$.

The chance-constrained optimization always ensures the satisfaction of the system $G z = d$, which may represent flow conservation constraints and other physical laws that cannot be violated. %in any circumstances. 
This setting, however, restricts the queries to be made only on strict subsets of solution $z$, because the perturbation needs to be redistributed among the variables to guarantee the feasibility of the system. By separating the solution into released and non-released variables, the latter can be optimized to provide trade-offs between the optimality loss and its variance, as well as the trade-offs between the optimality loss and the overall solution variance. The contributions of this work can be summarized as follows:
%
%\noindent\textbf{Contributions:} 
{\bf (1)} It develops a novel differentially private framework for the release of identity and sum queries over the solutions of constrained convex programs, using stochastic chance-constrained optimization (Section \ref{sec:algorithm}).
{\bf (2)} The released solutions are guaranteed feasible with a high probability. The feasibility guarantees are studied for individual and joint constraint satisfaction, providing higher or lower optimality losses, respectively (Section \ref{sec:ref_and_guarant}). 
{\bf (3)} The framework establishes a trade-off between the expected and the worst-case errors by controlling the variance of the optimization results and the optimality loss (Section \ref{sec:variance_control}).
{\bf (4)} On the benchmark energy optimization datasets, the framework is shown to outperform the standard output perturbation algorithm (Section \ref{sec:experiments}).  

\noindent\textbf{Notation} Upper and lower case symbols are used to denote, respectively, matrices and vectors. The indexed notation $A_i$ is used to denote the $i^{\text{th}}$ row vector of matrix $A$. The operator $\text{diag}(a)$ returns the diagonal matrix with entries of vector $a$, and $\text{diagv}(A)$ returns the vector of diagonal elements of matrix $A$. The ceil $\lceil r \rceil$ maps real number $r$ into the least succeeding integer. Notation $\circ$ denotes a Schur product. 
$\mathbb{0}$ and $\mathbb{1}$ respectively denote vectors of zeros and ones of proper dimensions.

\section{Preliminaries}\label{sec:preliminaries}
Across the paper, it is assumed that the optimal solution to problem \eqref{base_problem} exists and is unique and that the data $d_i$, contributed by each individual $i$, in $d$ are not correlated. Thus, the problem can be seen as an algorithm
$
\cM: \mathbb{R}^{\ell}\mapsto\mathbb{R}^{n}
$
with a unique mapping of datasets $d$ to optimization results. To enable private queries over the optimization results, this work considers a differentially private counterpart $\tilde{\mathcal{M}}$ of $\mathcal{M}$. While the traditional differential privacy definition aims at protecting the \emph{participation} of an individual data \citep{Dwork:06}, %for the constrained optimization problems of interest, 
this work focuses on obfuscating the \emph{magnitude} $d_i$ associated with participant $i$ of the input vector $d$. To capture this privacy notion, the paper focuses on the \emph{indistinguishability} framework proposed by \citet{chatzikokolakis2013broadening}, which protects the sensitive data of each individual up to some measurable quantity $\alpha > 0$ and defines two neighboring datasets $d, d'$ (written $\sim_\alpha$) as 
\begin{align*}
    d\sim_{\alpha}d' \Longleftrightarrow \exists i\;\text{s.t.}\;
    | d_{i}-d_{i}' | \leqslant \alpha \;\wedge\;d_{j}=d_{j}',\forall j \neq i,
\end{align*}
where $d$ and $d'$ are input vectors to problem \eqref{base_problem} and $\alpha$ is a positive real value. Following previous work \citep{munozprivate}, this relation requires the assumption that the neighboring datasets are feasible for problem \eqref{base_problem}, which is not restrictive, as only feasible solutions are of interest to release.

Differential privacy requires that the maximum divergence of the algorithm output distributions on neighboring inputs to be bounded by privacy parameters $\varepsilon$ and $\delta$, such that 
\begin{align*}
    \mathbb{P}\bracketit{\tilde{\mathcal{M}}(d) = O} \leqslant 
    \mathbb{P}\bracketit{\tilde{\mathcal{M}}(d')= O}\exp(\varepsilon) + \delta
\end{align*}
for a random algorithm $\tilde{\mathcal{M}}$ and any output $O$, where $\mathbb{P}$ denotes the probability over runs of $\tilde{\mathcal{M}}$.
If $\delta =0$, $\tilde{\mathcal{M}}$ is said to be $\varepsilon$-differentially private. 

The \emph{global sensitivity} methods are known to provide differential privacy by augmenting the output of computations with the noise  calibrated to the $\ell_1$- or $\ell_2$-sensitivity. 
The $\ell_{p}-$sensitivity 
$$\Delta_\alpha \triangleq \max_{d \sim_\alpha d'}\norm{\mathcal{M}(d)-\mathcal{M}(d')}_{p}, \quad p = 1,2,$$
is used to bound the change in the algorithm output induced by any two $\alpha$-indistinguishable inputs. 
In many applications of interest, 
$G\in\{0,1\}^{\ell\times n}$ is a binary matrix and the domain of datasets $d$ is normalized in $[0,1]$. Thus, $\Delta_{\alpha}$ is directly upper-bounded by $\alpha$.

%For constrained optimization problems, however, it is generally difficult to conclude how the sensitive inputs affect the optimization solution. \citet{han2014differentially} suggested approximating $\Delta_{\alpha}$ by the maximum diameter of the Löwner-John ellipsoid that contains a convex hull of $\mathcal{Z}$, however this approximation can be over-conservative and requires solving a large-scale semi-definite program with a limited computational tractability \citep[Section 3.7.2.2]{ben2001lectures}. Alternatively, the sensitivity $\Delta_{\alpha}$ can be reasoned from the problem design: if $G\in\{0,1\}^{\ell\times n}$ is a binary matrix, $\Delta_{\alpha}$ can be upper-bounded by $1\cdot\alpha$, provided that the domain of neighboring datasets $d$ is normalized in $[0,1]$. In the applications considered in this work, $\Delta_{\alpha}$ can be upper-bounded by $\alpha$. 

Let $\text{Lap}(\lambda)^n$ denote the i.i.d.~Laplace distribution over $n$ dimensions with $0$ mean and scale $\lambda$. The following ubiquitous result provides an $\varepsilon$-differentially private algorithm \citep{Dwork:06}. % under the $\ell_1$ sensitivity.
\begin{theorem}[Laplace mechanism]\label{th:laplace}
    Let $\mathcal{M}$ be an algorithm with $\ell_1$ sensitivity $\Delta_\alpha$ that maps datasets $d$ to $\mathbb{R}^{n}$. The Laplace mechanism $\mathcal{M}(d) + \xi$, with $\xi\sim\text{Lap}(\nicefrac{\Delta_\alpha}{\varepsilon})^{n}$, attains $\varepsilon$-differential privacy. 
\end{theorem}

% \section{Global sensitivity chance constrained optimization framework}\label{sec:algorithm}
\section{Internalizing global sensitivity methods into constrained optimization}\label{sec:algorithm}
% \nandoSide{Let's think at a title that stays in 1 line}
% \vladSide{Okay. I suggest this more general and self-containing title.}

The direct application of Theorem \ref{th:laplace} to the optimal optimization solution may produce a result that violates the problem constraints. This section introduces a suitable transformation of problem \eqref{base_problem} into a stochastic chance-constrained problem that internalizes the global sensitivity methods to establish both privacy and constraint feasibility guarantees. This section first discusses the Laplace mechanism and then extends the results to the Gaussian mechanism. 
% \footnote{While the results are discussed for the Laplace mechanism, the proposed method extends to provide $(\varepsilon,\delta)$-differential privacy using the Gaussian mechanism \citep{Dwork:06}.}
% \nandoSide{specify where the Gaussian mechanism is discussed in the Appendix} 

Consider a random perturbation $\xi\in\mathbb{R}^{p}$ calibrated to Laplace distribution $\mathbb{P}_{\xi} = \text{Lap}(\nicefrac{\Delta_{\alpha}}{\varepsilon})^p$ for some arbitrary dimension $p$, and assume that the solution $z$ depends on the realization of $\xi$ as
\begin{align}\label{rand_sol}
    z(\xi)= \tilde{z} + f(\xi) = \tilde{z} + Z\xi,
\end{align}
where $\tilde{z}\in\mathbb{R}^{n}$ is the expected value of the solution with respect to distribution $\mathbb{P}_{\xi}$ and $Z\xi$ is the linear functional recourse with recourse decision $Z\in\mathbb{R}^{n\times p}$, which is used to adjust the expected solution to any realization of $\xi$. Therefore, any query made over $z(\xi)$ will constitute the expected and random components. To provide privacy guarantees, the random component is required to be  independent from data $d$. This can be achieved by enforcing additional, query-specific, constraints $\mathcal{Q}$ on the recourse $Z$. While the framework can accommodate the general class of linear queries over solution $z(\xi)$, for ease of presentation, this work focuses on identity and sum queries.

\begin{definition}[Identity query]\label{def:identity_query}\normalfont
This query releases a specified subset of solution $z(\xi)$. 
Consider a random perturbation $\xi\in\mathbb{R}^{n}$ and a diagonal matrix $I\in\mathbb{R}^{n\times n}$, such that 
\begin{align*}
    \xi_{i} = \left\{
    \begin{array}{l}
        \text{Lap}(\nicefrac{\Delta_{\alpha}}{\varepsilon}), \\
        0,
    \end{array} \right.
    I_{ii} = \left\{
    \begin{array}{rl}
        1, & \text{if $i^{\text{th}}$ element of vector $z(\xi)$ is subject to release},\\
        0, & \text{otherwise.}
    \end{array} \right.
\end{align*}
The identity query release is thus $Iz(\xi) = I\tilde{z} + IZ\xi$ for $Z\in\mathbb{R}^{n\times n}$. 
The random component $IZ\xi$ is made independent from the dataset $d$ if the recourse decision $Z$ is constrained as follows
$$
Z\in\mathcal{Q}(I)\triangleq\{Z\;|Z\circ I = I,\;Z\circ(I-\text{diagv}(I)\mathbb{1}_{n}^{\top}) = \mathbb{0}_{n\times n}\},
$$
which yields $Iz(\xi)=\left\{
    \begin{array}{ll}
        \tilde{z}_{i} + \xi_{i}, & \text{if $I_{ii}=1$,}\\
        0, & \text{otherwise.}
    \end{array} \right.$
\end{definition}

\begin{definition}[Sum query]\label{def:sum_query}\normalfont This query releases $p$ sum statistics over non-intersecting subsets of $z(\xi)$.  Consider a random perturbation $\xi\in\mathbb{R}^{p}$ and a matrix $S\in\mathbb{R}^{p\times n}$, such that 
\begin{align*}
    S_{ij} = \left\{
    \begin{array}{rl}
        1, & \text{if element $z_{j}(\xi)$ participates in sum statistic $i$},\\
        0, & \text{otherwise.}
    \end{array} \right.
\end{align*}
The sum query releases $Sz(\xi)=S\tilde{z} + SZ\xi$ for $Z\in\mathbb{R}^{n \times p}$. 
The random component $SZ\xi$ is made independent from the dataset $d$ if the recourse decision $Z$ is constrained as follows
$$
Z\in\mathcal{Q}(S)\triangleq\{Z\;|SZ = \text{diag}(\mathbb{1}_{p})\},
$$
which yields $S z(\xi)= S \tilde{z} + \xi$.
%which yields $L_{j}^{\top} z(\xi)= L_{j}^{\top} \tilde{z} + \xi_{j},$ for all sum queries $j=1,\dots,g.$
\end{definition}

% Since multiple sum queries are made over non-intersecting subsets of solution $z$, the parallel composition guarantees that multiple answers do not induce additional privacy loss \citep{Dwork:06}.  

To produce random solutions to \eqref{base_problem}, function \eqref{rand_sol} is optimized using the following stochastic program
% The random solution \eqref{rand_sol} can be certified feasible for \eqref{base_problem} if the following stochastic problem
\begin{subequations}
\label{cc_problem_v2}
\begin{align}
        \minimize{\tilde{z},Z\in\mathcal{Q}}\quad &\mathbb{E}^{\mathbb{P}_\xi}\bracketit{c(\tilde{z} + Z\xi)}\label{cc_problem_v2_obj}\\
        \st\quad &\mathbb{P}_{\xi}\bracketit{A(\tilde{z} + Z\xi)\leqslant b}\geqslant 1-\eta,\label{cc_problem_v2_cc}\\
        &G(\tilde{z} + Z\xi)=d \quad\mathbb{P}_{\xi}\text{-a.s.},\label{cc_problem_v2_asc}
\end{align}
\end{subequations}
which optimizes $\tilde{z}$ and $Z$ by anticipating all realizations of the random variable $\xi$. This problem minimizes the expected value of the convex cost function \eqref{cc_problem_v2_obj} with respect to the random variable $\xi$.
% for the given distribution $\mathbb{P}_\xi$ by optimizing decisions $\tilde{z}$ and $Z$ subject to the query-specific constraints in $\mathcal{Q}$. Note, that the expectation is taken over the convex cost function range, not with respect to the feasible range. 
The problem constraints are given by a set of probabilistic constraints. The joint chance constraint \eqref{cc_problem_v2_cc} requires the satisfaction of the inequality constraints with a {\it prescribed} probability $1-\eta$, specified by the curator of problem \eqref{cc_problem_v2}. The almost sure constraint \eqref{cc_problem_v2_asc} requires the equality constraints to hold with probability 1. Note that, if problem \eqref{cc_problem_v2} is infeasible, it follows that the privacy parameters $\alpha$ and $\varepsilon$ are too strong for the feasibility requirement $\eta$. 

As the recourse of problem \eqref{cc_problem_v2} amounts to finitely-dimensional linear functions, objective function \eqref{cc_problem_v2_obj} and chance constraint \eqref{cc_problem_v2_cc} admit computationally tractable reformulations \citep{ben2009robust} (additional details will be given in Section \ref{sec:ref_and_guarant}). 
The almost sure constraint \eqref{cc_problem_v2_asc} includes a random variable and, therefore, satisfying it is computationally intractable. However, it can be equivalently reformulated using the following set of equations:
\begin{align}\label{balance}
G\tilde{z} = d,\quad GZ = \mathbb{0}.
\end{align}
If variables $\tilde{z}$ and $Z$ are subject to \eqref{balance}, their optimal solution satisfies the equality constraint \eqref{cc_problem_v2_asc} for any realization of $\xi$. The structural properties of $G$ restrict the set of potential queries and a query is said to be implementable if there exists $Z\in\mathcal{Q}$ such that $GZ=\mathbb{0}$ holds.  

\begin{example*}[Flow conservation constraint]\normalfont
Assume $G\in\mathbb{R}^{\ell\times n}$ represents the incidence matrix of a fully connected graph $\mathcal{G}$ (its rank is $n-1$) and that $G z = d$ represents a flow conservation constraint. Consider an identity query $Iz(\xi) = I\tilde{z} + IZ\xi$ and $Z\in \mathcal{Q}(I)$ as in Definition \ref{def:identity_query}. The identity query is implementable if $\Tr[I] < n$, i.e., not all elements of $Z$ are constrained by $\mathcal{Q}(I)$ and $GZ = \mathbb{0}$ holds.
\end{example*}
The constraint $GZ = \mathbb{0}$ plays a critical role: it balances the perturbation between the released and non-disclosed variables. As a result, the query should leave enough degree of freedom to satisfy the equality constraint.
%
% Recall that the queries act on the subset of the solution, and the presence of equality constraints158hinders the release of queries over theentirevectorz(ξ).
This limitation is solely induced by the need to preserve the satisfaction of the equality constraint and it is not seen as a limiting factor for many applications
% , including the one discussed in this work 
(see Section \ref{sec:experiments}). 

\paragraph{Private identity query (PIQ) algorithm} The procedure is summarized in Algorithm \ref{alg:identity}, which takes as inputs the dataset $d$, the $\ell_{1}$-sensitivity of the identity query, the privacy $\varepsilon$ and feasibility $\eta$ requirements, the known covariance $\Sigma$, and the query specification $I$. Upon receiving the optimal chance-constrained solution (line 2), the algorithm draws a sample from the Laplace distribution (line 3) and computes a $(1-\eta)$-feasible solution for problem \eqref{base_problem} (line 4). 
The algorithm returns an $\varepsilon$-differentially private identity query which satisfies problem \eqref{base_problem} constraints with probability $(1-\eta)$.

\begin{theorem}[$\varepsilon$-differentially PIQ]\label{th:identity_query}
Algorithm \ref{alg:identity} is $\varepsilon$-differentially private, i.e., 
\begin{align*}
    \mathbb{P}\bracketit{I\braceit{\optimal{z}(d) + \optimal{Z}(d)\xi} = O}\leqslant\mathbb{P}\bracketit{I\braceit{\optimal{z}(d') + \optimal{Z}(d')\xi} = O}\text{exp}(\varepsilon)
\end{align*}
for any two $\alpha-$neighboring datasets $d$ and $d'$ and output solutions $O$. %The probability $\mathbb{P}$ is taken over runs of Algorithm \ref{alg:identity}.
\end{theorem}

\paragraph{Private sum query (PSQ) algorithm} The procedure is summarized in Algorithm \ref{alg:linear}, which differs from Algorithm \ref{alg:identity} by the query specification $S$ and the noise dimension.  

\begin{theorem}[$\varepsilon$-differentially PSQ]\label{th:affine_transofrmation_query}
Algorithm \ref{alg:linear} is $\varepsilon$-differentially private, i.e.,  
\begin{align}
    \mathbb{P}\bracketit{S(\optimal{z}(d)+\optimal{Z}(d))\xi = O}\leqslant\mathbb{P}\bracketit{S(\optimal{z}(d')+\optimal{Z}(d'))\xi = O}\text{exp}(\varepsilon)
\end{align}
for any two $\alpha-$neighboring datasets $d$ and $d'$ and output solutions $O$. %The probability $\mathbb{P}$ is taken over runs of Algorithm \ref{alg:linear}.
\end{theorem}

\begin{figure}[tp]
\centering
\begin{minipage}{0.47\textwidth}
\begin{algorithm}[H]
        {\bf Input:} $d,\Delta_\alpha, \varepsilon, \eta, I$\\
        $(\optimal{\tilde{z}}\;, \optimal{Z}\;) \gets $
        Solve \eqref{cc_problem_v2} for $Z\in\mathcal{Q}(I)$\\
        $\hat{\xi} \gets$ Sample from $\text{Lap}(\nicefrac{\Delta_\alpha}{\varepsilon})^{n}$\\
        Compute solution to \eqref{base_problem} as $\hat{z}=\optimal{\tilde{z}}\;+\optimal{Z}\hat{\xi}$\\
         \bf{Release:} $I\hat{z}$ 
        % \leIf{$A\hat{z}\leqslant b$}{
        %    \bf{Release:} $I\hat{z}$\\
        %    }{
        %    \bf{Release:} $\bot$ 
        % }
        \caption{Private identity query (PIQ)}
        \label{alg:identity}
        %         {\bf Input:} $d,\Delta,\alpha,\varepsilon,\eta$\\
        %         Solve \eqref{cc_problem_identity} for $\xi\sim\text{Lap}(\frac{\alpha\Delta}{\varepsilon})^{k}$\\
        %         Sample $\hat{\xi}$ from $\text{Lap}(\frac{\alpha\Delta}{\varepsilon})^{k}$\\
        %         {\bf Release:} $\hat{x}=\optimal{x} + \optimal{X}\hat{\xi}$
\end{algorithm}
\end{minipage}
\hspace{0.5cm}\begin{minipage}{0.47\textwidth}
\begin{algorithm}[H]
        {\bf Input:} $d,\Delta_\alpha, \varepsilon, \eta,  S$\\
        $(\optimal{\tilde{z}}\;, \optimal{Z}\;) \gets $
        Solve \eqref{cc_problem_v2} for $Z\in\mathcal{Q}(S)$\\
        $\hat{\xi} \gets$ Sample from $\text{Lap}(\nicefrac{\Delta_\alpha}{\varepsilon})^{p}$\\
        Compute solution to \eqref{base_problem} as $\hat{z}=\optimal{\tilde{z}}\;+\optimal{Z}\hat{\xi}$\\
        \bf{Release:} $S\hat{z}$
        % \leIf{$A\hat{z}\leqslant b$}{
        %    \bf{Release:} $S\hat{z}$\\
        %    }{
        %    \bf{Release:} $\bot$ 
        % }
        % {\bf Input:} $d,\Delta,\alpha,\varepsilon,\eta,\Gamma$\\
        % Solve \eqref{cc_problem_linear} for $\xi\sim\text{Lap}(\frac{\alpha\Delta}{\varepsilon})^{1}$\\
        % Sample $\hat{\xi}$ from $\text{Lap}(\frac{\alpha\Delta}{\varepsilon})^{1}$\\
        % {\bf Release:} $\hat{x}=\mathbb{1}^{\top}\Gamma(\optimal{x}-\optimal{X})\hat{\xi}$
        \caption{Private sum query (PSQ)}
        \label{alg:linear}
\end{algorithm}
\end{minipage}
\end{figure}
%%%%%%%%%%%%
In addition to releasing a privacy-preserving answer with a probabilistic feasibility certificate, the proposed framework also allows to verify the feasibility of the sampled solution $\hat{z}$ without incurring an additional privacy loss. Since the equality constraint holds due to \eqref{balance}, it is sufficient to verify the feasibility of constraint $A\hat{z}\leqslant b$ without accessing the original data $d$. The operation is private by post-processing immunity of differential privacy \citep{dwork2014algorithmic}.

Furthermore, since formulation \eqref{cc_problem_v2} is independent from the distribution of the noise, the framework can accommodate other global sensitivity methods. In particular, the following result holds. 
\begin{theorem}[Gaussian algorithms]\label{th:gaussian_mech} Let $\delta,\varepsilon\in(0,1)$ and let $\Delta_{\alpha}^{2}$ be the $\ell_{2}-$sensitivity.  
Algorithms \ref{alg:identity} and \ref{alg:linear} that calibrate $\xi\in\mathcal{N}(0,\sigma^2)$ to the Gaussian distribution with $\sigma\geqslant \Delta_{\alpha}^{2} \sqrt{2\ln(1.25/\delta)}/\varepsilon$ satisfy $(\varepsilon,\delta)$- differential privacy.
\end{theorem}

The relation between the feasibility requirement $\eta$ and the privacy parameters $\varepsilon$ and $\delta$ is implicit in the formulation of the chance-constrained problem: the variance of the noise affects the ability to satisfy the problem constraints within the feasibility requirement and vice-versa. To render this relation explicit, Appendix \ref{composition_algorithms} discusses a version of Algorithms \ref{alg:identity} and \ref{alg:linear} that iterates lines 3 to 5 an optimal number $T$ of times to guarantee the release of a feasible solution with probability $1-\mu$, for some $0 < \mu < 1$.

\begin{theorem}[Composition to improve feasibility]\label{th:composition}
% Given feasibility requirement $\eta$, privacy parameter $\varepsilon/T$, and value $0 < \mu < 1$, the \emph{iterative} variants of Algorithms \ref{alg:identity} and \ref{alg:linear} return a feasible and $\varepsilon$-differentially private solution within $T = \lceil \frac{\log(\mu)}{\log(\eta)} \rceil$ steps with probability at least $1-\mu$.
Given feasibility requirement $\eta$, privacy parameter $\varepsilon/T$, and value $0 < \mu < 1$, the \emph{iterative} variants of Algorithms \ref{alg:identity} and \ref{alg:linear} return an $\varepsilon$-differentially private solution that is feasible with probability at least $1-\mu$ within $T = \lceil \frac{\log(\mu)}{\log(\eta)} \rceil$ steps.
\end{theorem}

\section{Reformulations and feasibility guarantees}\label{sec:ref_and_guarant}
The optimization problem \eqref{cc_problem_v2} is intractable because it constitutes the optimization of a random variable. However, due to the convexity assumption on \eqref{base_problem}, linear functional recourse and known distribution of $\xi$, problem \eqref{cc_problem_v2} admits tractable reformulations. There are several avenues to reformulate the joint chance constraint \eqref{cc_problem_v2} with different degrees of conservatism in terms of expected optimality loss \citep{nemirovski2007convex}. This work provides a conservative joint constraint satisfaction guarantee, using a \emph{sample} approximation, and a less conservative individual constraint satisfaction guarantee, using an \emph{analytic} reformulation. The objective function is reformulated as follows.

\paragraph{Objective function reformulation} Consider a quadratic cost function with first- and second-order coefficients $c_{1}\in\mathbb{R}^{n}$ and $c_{2}\in\mathbb{R}^{n}$, and a \emph{diagonal} covariance matrix $\Sigma=\mathbb{E}[\xi\xi^{\top}]$ with diagonal elements being equal to $\lambda=\Delta_\alpha/\varepsilon$. Then, the objective function \eqref{cc_problem_v2_obj} reformulates as
\begin{align*}
    &\mathbb{E}^{\mathbb{P}_\xi}\bracketit{c_{1}^{\top}(z+Z\xi) + (z+Z\xi)^{\top}\text{diag}(c_{2})(z+Z\xi)}
    = c_{1}^{\top}z + z^{\top}\text{diag}(c_{2})z + \Tr\bracketit{Z^{\top}\text{diag}(c_{2})Z\Sigma},
\end{align*}
which follows from the zero-mean distribution $\mathbb{P}_{\xi}$ and the fact that $\mathbb{E}[\xi\xi^{\top}]=\Sigma.$ Notice that for the affine cost functions, the analytic reformulation of \eqref{cc_problem_v2_obj} reduces to $c_{1}^{\top}z$.

\paragraph{Sample approximation} This approximation substitutes the chance constraint \eqref{cc_problem_v2_cc} with a finite number of deterministic constraints, each enforced on a specific realization of random perturbation \citep{campi2008exact,alamo2010sample,margellos2014road}. This work invokes the sample approximation method from \citep{margellos2014road}, which enforces \eqref{cc_problem_v2_cc} on the vertices of the rectangular sample set extracted from distribution $\mathbb{P}_{\xi}$, i.e., for $\xi\in\mathbb{R}^{p}$
\begin{align}\label{sample_ref}
    \mathbb{P}_{\xi}\bracketit{
    \braceit{A(z+Z\xi)\leqslant b}}\geqslant 1-\eta
    \quad\equiv\quad
    Az \leqslant b - AZ\hat{\xi}^{v},\quad\forall v=1,\dots,2^{p},
\end{align}
where $\hat{\xi}^{v}\in\mathbb{R}^{2^p}$ is the $v^{\text{th}}$ vertex of the extracted sample set. \citet{margellos2014road} show that the joint constraint satisfaction is attained if the number of samples $N$ from $\mathbb{P}_{\xi}$ is properly chosen. 
\begin{theorem}[\citet{margellos2014road}]
    The equivalence \eqref{sample_ref} holds with confidence $(1-\beta)$ if the rectangular set is built upon $S$ samples extracted from $\mathbb{P}_{\xi}$, with $N$ at least as much as 
    $$
    %\textstyle
    N \geqslant \left\lceil
    \frac{1}{\eta}\frac{e}{e-1}\braceit{2p-1+\text{ln}\frac{1}{\beta}}
    \right\rceil.
    $$
    %where $e$ is Euler’s number.
\end{theorem}
Finally, notice that this approximation requires an additional input $\beta$ to Algorithms \ref{alg:identity} and \ref{alg:linear} to accommodate the confidence level of the data curator. 

\paragraph{Analytic reformulation} 

The joint chance constraint \eqref{cc_problem_v2_cc} can be rewritten as a union of individual chance constraints. For some vector $\overline{\eta}\in\mathbb{R}_{+}^{m}$ of individual constraint violation probabilities, the individual chance constraints can be reformulated exactly using second-order cone constraints \citep{ben2001lectures}:
\begin{align}\label{analytic_ref}
%     \textstyle
     A_{i}z \leqslant b_{i} - f(1-\overline{\eta}_{i})\norm{A_{i}Z\Sigma^{\nicefrac{1}{2}}}_{2}, \quad\forall i=1,\dots,m,
\end{align}
where $f(1-\overline{\eta}_{i})$ is a distribution-dependent safety parameter and $\Sigma^{\nicefrac{1}{2}}$ is the lower triangular matrix resulting from the Cholesky factorization of $\Sigma$. For any symmetric and unimodal distribution of $\xi$, $f(1-\overline{\eta}_{i})$ amounts to $(\nicefrac{2}{9\overline{\eta}_{i}})^{\nicefrac{1}{2}}$ if $0\leqslant\overline{\eta}_{i}\leqslant\nicefrac{1}{6}$ (see the result from \cite{van2016generalized}). For the Gaussian distribution of $\xi$, $f(1-\overline{\eta}_{i})$ amounts to the inverse CDF of the standard Gaussian distribution at $(1-\overline{\eta}_{i})-$quantile \citep{ben2001lectures}. Observe that, for the fixed parameter $\overline{\eta}_{i}$, the last term in the right-hand side of \eqref{analytic_ref} is a safety margin, which reduces the feasible space of the original problem \eqref{base_problem} to guarantee individual constraint feasibility for $(1-\overline{\eta}_{i})$ realizations of $\xi$. 

If $\mathbb{1}^{\top}\overline{\eta}\leqslant\eta$ holds, the individual chance constraints guarantee joint constraint satisfaction probability $\eta$. Yet, finding the optimal value $\overline{\eta}$ is an NP- hard problem \citep{xie2019optimized}. Section \ref{sec:experiments}, shows that, for small problem instances, the choice $\overline{\eta}=\eta$ results in the desired joint constraint satisfaction while providing a significantly less conservative solution than the sample approximation.

\section{Variance-aware differentially private algorithms}\label{sec:variance_control}

For many systems governed by the solution of problem \eqref{base_problem}, e.g. energy networks, it is important to control the impact of the differentially private solutions on the optimality loss (e.g., extra supply cost) and the variance of the state variables (e.g., supply and flow allocations). This section extends Algorithms \ref{alg:identity} and \ref{alg:linear} to provide a minimal variance solution without affecting the privacy guarantees.  

\paragraph{Minimal variance of optimality loss} 
The chance-constrained problems of Algorithms \ref{alg:identity} and \ref{alg:linear} optimize against the expected value of the cost function. Its solution provides the estimate of the expected optimality loss relative to the solution of problem \eqref{base_problem}. The worst-case outcome of the optimality loss, however, may significantly exceed the expected value. A trade-off between the expected and worst-case outcomes can be attained by controlling the variance of the optimality loss. 

As the value of the cost function of problem \eqref{base_problem} is deterministic, it is sufficient to control the variance of \eqref{cc_problem_v2_obj} to attain the desired result. For a linear cost function, the variance admits a convex expression $\text{Var}\bracketit{c_{1}^{\top}(z + Z\xi)}=\Tr\bracketit{Z^{\top}\text{diag}(c_{1})\text{diag}(c_{1})Z\Sigma}$ in recourse variable $Z$. Therefore, it can be minimized by optimizing,  instead, the following objective function
\begin{align}\label{var_problem_loss}
        \minimize{z,Z\in\mathcal{Q}}\quad& (1-\varphi)\mathbb{E}^{\mathbb{P}_\xi}\bracketit{c_{1}^{\top}(z + Z\xi)} + \varphi\norm{\Sigma^{\nicefrac{1}{2}}Z^{\top}c_{1}}_{2},
\end{align}
which optimizes the trade-off between the expected value and the standard deviation of the cost function for some factor $\varphi\in[0,1]$. 
Thus, varying the factor $\varphi$ establishes a Pareto frontier between the optimality loss and its variance. Since the recourse decision $Z$ is subject to query-specific constraints $\mathcal{Q}$, the results of Theorems \ref{th:identity_query} and \ref{th:affine_transofrmation_query} hold. Finally, the variance of the non-affine cost functions does not permit convex formulations and is not considered in this paper. 

\paragraph{Minimal variance of optimization variables}  The variance of the  optimization solution $z\braceit{\xi}=z+Z\xi$ admits a convex expression $\text{Var}\bracketit{z\braceit{\xi}} = \Tr\bracketit{Z^{\top}\Sigma Z}$ in $Z$. Therefore, it can be controlled by optimizing the recourse decision $Z$ using the following objective function
\begin{align}\label{var_problem_variables}
        \minimize{z,Z\in\mathcal{Q}}\quad& (1-\varphi)\mathbb{E}^{\mathbb{P}_\xi}\bracketit{c(z + Z\xi)} + \varphi\norm{Z\Sigma^{\nicefrac{1}{2}}\mathbb{1}}_{2},
\end{align}
which finds the optimal trade-off between the expected cost and the standard deviation of the optimization variables for some factor $\varphi\in[0,1]$. Since the optimal recourse is still guided by the query-specific constraints, the privacy guarantees provided by Theorems \ref{th:identity_query} and \ref{th:affine_transofrmation_query} are preserved. 

%%%%%%%%%%
% 
%%%%%%%%%%
\section{Experiments}
\label{sec:experiments}
\DefineShortVerb{\#}% # denotes verbatim opening/closing character
\SaveVerb{OP}#OP#
\SaveVerb{PIQ-a}#PIQ-a#
\SaveVerb{PIQ-s}#PIQ-s#
\SaveVerb{PLQ-a}#PSQ-a#
\SaveVerb{PLQ-s}#PSQ-s#

\paragraph{Problem description} 
The proposed framework is applied to the energy resource allocation problem using a set of benchmark networks from \citep{coffrin2018powermodels}. The problem goal is to compute the cost-optimal supply allocations across the network to satisfy nodal demands while respecting the supply and network limits. The problem is described by an undirected graphs $\mathcal{G}(N,E)$ with a set of nodes $N$ and a set of edges $E$, connecting those nodes. The graph typology is represented by the weighted Laplacian matrix $B$ formed from non-negative edge weights $\beta\in\mathbb{R}_{+}^{\setcard{E}}$. The nodal supply  $p\in\mathbb{R}_{+}^{\setcard{N}}$ is allocated in the network to meet nodal demand $d\in\mathbb{R}_{+}^{\setcard{N}}$. The flow along the edges is modeled considering a vector of nodal potentials $\theta\in\mathbb{R}^{\setcard{N}}$, their difference is proportional to the network flows, i.e., the flow in edge $\ell$ amounts to $f_{\ell}(\theta)=\beta_{\ell}(\theta_{\se{\ell}}-\theta_{\re{\ell}}),\forall \ell\in E$, with operators $\se{\ell}$ and $\re{\ell}$ returning the sending and receiving nodes of edge $\ell$, respectively. Finally, the nodal supply incurs costs computed by function $c:\mathbb{R}^{\setcard{N}}\mapsto\mathbb{R}$. This allocation problem gives rise to the following optimization
\begin{subequations}
\label{problem_det}
\begin{align}
    \minimize{p,\theta} \quad & c(p)\\
    \st \quad& B\theta = p - d \label{problem_det_bal}\\  
             &\underline{p} \leqslant p \leqslant \overline{p} \label{problem_det_supply}\\  
             &\underline{f} \leqslant f(\theta) \leqslant \overline{f}. \label{problem_det_flow}
\end{align} 
\end{subequations}
The objective function minimizes the total supply cost, while the equality constraint balances nodal demand, supply, and net flow injection. The inequality constraints respect the minimum and maximum nodal supply and the network flow limits $\underline{p},\overline{p}\in\mathbb{R}_{+}^{\setcard{N}}$, $\underline{f}\in\mathbb{R}_{-}^{\setcard{E}},\overline{f}\in\mathbb{R}_{+}^{\setcard{E}}$. A rearrangement of the terms in \eqref{problem_det_bal}-\eqref{problem_det_flow} makes the problem representable in the form expressed by problem \eqref{base_problem}, thus its chance-constrained counterpart is achieved as detailed in Sections \ref{sec:algorithm} and \ref{sec:ref_and_guarant}.

The experiments concern the identity and sum queries made over the subset of nodal supplies $p$ and make use of the $\ell_{1}-$sensitivity $\Delta_{\alpha}$ of vector $p$ on the two $\alpha-$indistinguishable datasets $d$ and $d'$. Consider the optimal supply allocations $\optimal{p}\;$ and $\optimal{p}'$ obtained, respectively, on datasets $d$ and $d'$.
\begin{proposition}
    $\Delta_{\alpha}=\norm{\optimal{p}-\optimal{p}'}_{1} \leqslant \alpha.$
\end{proposition}
\begin{proof}
The equality constraint \eqref{problem_det_bal} requires the balance between the total supply and total demand. By construction, $\sum_{i\in N}[B\optimal{\theta}\;]_{i} = 0,$ thus from \eqref{problem_det_bal} it follows
\begin{align*}
    &\sum_{i\in N} \optimal{p}_{i} - \sum_{i\in N} d_{i} = 0 \\
    &\sum_{i\in N} \optimal{p}_{i}^{\prime} - \sum_{i\in N} d_{i}^{\prime} = 0,
\end{align*}
therefore,
$$
\sum_{i\in N}\optimal{p}_{i} - \sum_{i\in N}\optimal{p}_{i}^{\prime} = \sum_{i\in N} d_{i} - \sum_{i\in N} d_{i}^{\prime} \leqslant \alpha,
$$
because the datasets $d$ and $d'$ differ by at most $\alpha$ in one entry, i.e., $\norm{d-d'}_{1}\leqslant\alpha$. 
Therefore, $\Delta_{\alpha}=\norm{\optimal{p}-\optimal{p}'}_{1} \leqslant \alpha$. 
\end{proof}

\paragraph{Experimental setup}
The experiments are organized as follows. For every network instance, the variable limits are fixed, while cost coefficients and nodal demands are i.i.d.~drawn from the following uniform distributions  $c_{1}\sim U(1,3)$, $c_{2}\sim U(\nicefrac{1}{10},\nicefrac{3}{10})$, and $d\sim U(\nicefrac{1}{2},1).$ The results are thus reported for 100 independent simulation runs. The identity and sum queries are made over an arbitrary set of 30\% of nodal supplies, which is sampled at every simulation run. The privacy loss parameter $\varepsilon$ is set to $1$ and the indistinguishability parameter $\alpha$ is set to $0.1$ for identity queries and $0.5$ for sum queries. As $\ell_{1}$-sensitivity $\Delta_{\alpha}$ is bounded by $\alpha$, random perturbations thus obey the Laplace distribution $\text{Lap}(\alpha)$. The feasibility requirements for the joint and individual constraint satisfaction are set uniformly at $\eta=\overline{\eta}=2.5\%$, and the out-of-sample empirical constraint violation probability is obtained for 1000 samples at every simulation run.

\paragraph{Privacy-preserving algorithms} The algorithm abbreviations \verb PIQ \; and \verb PSQ \; are appended by \verb -a \; or \verb -s \; to indicate whether the chance constraints are reformulated, respectively, analytically or by samples (Section \ref{sec:ref_and_guarant}). They are compared with the output perturbation \verb OP \;algorithm, which adds noise to the query answer. The \verb OP \;solution is said to be feasible if problem \eqref{base_problem} returns a feasible solution for the fixed solution of the \verb OP \;algorithm. 

\paragraph{Implementation} The simulations were carried out using the standard PC with Intel Core i5 3.4 GHz processor and 8 GB memory. Solving optimization problems with the analytic reformulation requires less than a few seconds on average, whereas the sample approximation of the chance constraints requires by at most 78 seconds on average. The optimization models were implemented in the Julia Language and the source code can be accessed at \url{https://github.com/wdvorkin/DP_CO_FG}.

\paragraph{Algorithm comparison} The algorithms are compared in terms of their ability to release private queries while satisfying the feasibility requirement. Table \ref{tab:identity} summarizes the results for identity query answers obtained on several networks differing by the number of variables ($n$) and constraints ($|\mathcal{Z}|$). The results indicate that the \verb OP \;algorithm returns private answers that violate the problem constraints at a far greater rate that the one imposed by the requirement $\eta$. Additionally, its performance degrades with the increase of the problem size.  The application of \verb PIQ-a, on the other hand, provides formal guarantees for the individual constraint satisfaction.  These guarantees suffice to attain the desired feasibility requirement for smaller network instances. However, with an increasing network size, the probability of violating multiple constraints increases, and the average \verb PIQ-a \; feasibility performance reduces.  To guarantee the joint constraint satisfaction within the prescribed probability $1-\eta$, the \verb PIQ-s \; uses a sample approximation. While the guarantees are attained, notice that (last four columns of the table) this algorithm generates solutions with larger optimality loss than those generated by \verb PIQ-a , reflecting the discussion in Section \ref{sec:variance_control}. 

Next, Table \ref{tab:linear} reports the results for the  sum queries made over the largest test case \verb 118_ieee. These sum queries return a single aggregated statistic for the subset of selected nodes (first row of the table), or return the sums over 3, 6, or 9 partitions of the selected nodes. A single statistic requires one perturbation, which is accommodated by all algorithms in a feasible manner. With an increasing number of statistics, however, the differences between \verb OP \;and \verb PSQ \;algorithms are clearly observed.

\paragraph{Variance-aware differentially private optimization} 
The last experiments show the ability of the chance-constrained framework to control the variance of the optimization results by means of Equations \eqref{var_problem_loss} and \eqref{var_problem_variables}. 
Due to the inherent dependency between the optimality loss and the cost values, the results are given for different degrees of sparsity ($\bar{c}$) of supply cost among network nodes. Figure \ref{figure_var} illustrates the results for the \verb PSQ-a \;algorithm releasing $9$ sum statistics for the \verb 118_ieee \hspace{0.1cm}case and for the various assignments of the trade-off parameter $\varphi\in[0,1]$. The left plot shows that independently of cost sparsity, the algorithm can produce a differentially private output at zero variance of the optimality loss. The right plot demonstrates the drastic reduction of the overall solution variance $\text{Var}\bracketit{z(\xi)}=\Tr\bracketit{Z^{\top}\Sigma Z}$, almost to the $\Tr[\Sigma]$ total variance of the $9$ random perturbations. Both results, however, require larger conservatism of the solution in terms of optimality loss.

\begin{table}
\centering

\caption{Identity query summary for 100 network data samples} 

\label{tab:identity}

\begin{adjustbox}{max width=0.9\textwidth}
\begin{tabular}{llcccccccccc}
\toprule
\multicolumn{1}{c}{\multirow{3}{*}{Case ID}} & \multirow{3}{*}{$n\times|\mathcal{Z}|$} & \multicolumn{6}{c}{Empirical constraint violation $\mathbb{P}[\hat{z}\notin\mathcal{Z}]$ {[}\%{]}} & \multicolumn{4}{c}{Optimality loss $\Delta c$ {[}\%{]}} \\
\cmidrule(r){3-8} \cmidrule(r){9-12}
\multicolumn{1}{c}{} &  & \multicolumn{2}{c}{ \UseVerb{OP} } & \multicolumn{2}{c}{\UseVerb{PIQ-a}} & \multicolumn{2}{c}{\UseVerb{PIQ-s}} & \multicolumn{2}{c}{\UseVerb{PIQ-a}} & \multicolumn{2}{c}{\UseVerb{PIQ-s}} \\
\cmidrule(r){3-4}\cmidrule(r){5-6}\cmidrule(r){7-8}\cmidrule(r){9-10}\cmidrule(r){11-12}
\multicolumn{1}{c}{} &  & mean & std & mean & std & mean & std & mean & std & mean & std \\
\midrule
\verb 3_lmbd  & 6$\times$17 & 29.7 & 23.07 & 0.64 & 0.34 & 0.27 & 0.31 & 3.42 & 3.55 & 5.72 & 7.64 \\
\verb 5_pjm   & 10$\times$29 & 18.32 & 22.94 & 0.39 & 0.41 & 0.12 & 0.3 & 1.22 & 2.07 & 2.04 & 3.03 \\
\verb 14_ieee & 28$\times$84 & 52.24 & 26.85 & 1.48 & 0.78 & 0.27 & 0.25 & 1.55 & 1.33 & 3.45 & 2.88 \\
\verb 39_epri & 78$\times$211 & 95.56 & 4.69 & 4.86 & 1.32 & 0.49 & 0.35 & 2.17 & 0.81 & 4.7 & 1.68 \\
\verb 57_ieee  & 114$\times$333 & 98.59 & 1.91 & 7.17 & 1.50 & 1.28 & 1.06 & 2.4 & 0.78 & 5.51 & 5.06 \\
\verb 118_ieee & 236$\times$728 & 99.99 & 0.02 & 14.35 & 2.06 & 1.51 & 0.47 & 2.46 & 0.56 & 4.89 & 1.20 \\
\bottomrule
\end{tabular}
\end{adjustbox}
\end{table}

\begin{table}
\centering

\caption{Sum query summary for 100 network data samples} 

\label{tab:linear}

\begin{adjustbox}{max width=0.72\textwidth}
\begin{tabular}{ccccccccccc}
\toprule
\parbox[t]{3mm}{\multirow{3}{*}{\rotatebox[origin=c]{90}{queries \#}}} & \multicolumn{6}{c}{Empirical constraint violation  $\mathbb{P}[\hat{z}\notin\mathcal{Z}]$ {[}\%{]}} & \multicolumn{4}{c}{Optimality loss $\Delta c$ {[}\%{]}} \\
\cmidrule(r){2-7}\cmidrule(r){8-11}
 & \multicolumn{2}{c}{ \UseVerb{OP} } & \multicolumn{2}{c}{\UseVerb{PLQ-a}} & \multicolumn{2}{c}{\UseVerb{PLQ-s}} & \multicolumn{2}{c}{\UseVerb{PLQ-a}} & \multicolumn{2}{c}{\UseVerb{PLQ-s}} \\
\cmidrule(r){2-3}\cmidrule(r){4-5}\cmidrule(r){6-7}\cmidrule(r){8-9}\cmidrule(r){10-11}
 & mean & std & mean & std & mean & std & mean & std & mean & std \\
\midrule
1 & 0.00 & 0.00 & 0.00 & 0.00 & 0.00 & 0.00 & 0.00 & 0.00 & 0.00 & 0.00 \\
3 & 0.62 & 0.03 & 0.00 & 0.00 & 0.00 & 0.00 & 0.07 & 0.04 & 0.12 & 0.11 \\
6 & 16.48 & 20.2 & 1.39 & 0.84 & 0.2 & 0.24 & 0.58 & 0.32 & 1.81 & 0.82 \\
9 & 59.83 & 23.6 & 4.94 & 1.11 & 1.00 & 0.92 & 2.36 & 1.01 & 13.35 & 7.31 \\
\bottomrule
\end{tabular}
\end{adjustbox}
\end{table}

\begin{figure}[!h]
\center
    \begin{tikzpicture}[thick,scale=1,font=\scriptsize]
    \begin{axis}[
    font=\scriptsize,
    xlabel near ticks, 
    ylabel near ticks, 
    ytick={0, 2.5, 5, 7.5},
    y tick label style={/pgf/number format/.cd,fixed,fixed zerofill,precision=1,/tikz/.cd},
    ylabel={$\text{Var}\bracketit{\loss}$}, 
    xlabel={$\mathbb{E}\bracketit{\loss}$}, 
    label style={font=\footnotesize},
    tick label style={font=\scriptsize}, 
    legend pos = south west, 
    legend style={font=\tiny, 
    legend cell align={left},fill=none,xshift=4.55cm,yshift=0cm}, 
    width=5.6cm, height=3.5cm]
        \addplot [black, line width=0.15mm, mark=pentagon*,every mark/.append style={solid, fill=gray}, thick,mark size=1.5pt,
            mark repeat*=20, smooth] 
        table [x index = 1, y index = 2] {data_plots/summary_loss.dat};
        \addlegendentry{$\overline{c}=10$};
        \addplot [black, line width=0.15mm, mark=*,every mark/.append style={solid, fill=gray}, thick,mark size=1.5pt,
            mark repeat*=20, smooth] 
        table [x index = 3, y index = 4] {data_plots/summary_loss.dat};
        \addlegendentry{$\overline{c}=20$};
        \addplot [black, line width=0.15mm, mark=otimes*,every mark/.append style={solid, fill=gray}, thick,mark size=1.5pt,
            mark repeat*=20, smooth] 
        table [x index = 5, y index = 6] {data_plots/summary_loss.dat};
        \addlegendentry{$\overline{c}=30$};
        \addplot [black, line width=0.15mm, mark=triangle*,every mark/.append style={solid, fill=gray}, thick,mark size=1.5pt,
            mark repeat*=20, smooth] 
        table [x index = 7, y index = 8] {data_plots/summary_loss.dat};
        \addlegendentry{$\overline{c}=40$};
        \addplot [black, line width=0.15mm, mark=square*,every mark/.append style={solid, fill=gray}, thick,mark size=1.5pt,
            mark repeat*=20, smooth] 
        table [x index = 9, y index = 10] {data_plots/summary_loss.dat};
        \addlegendentry{$\overline{c}=50$};
        \addplot[->,>=stealth, gray] coordinates
           {(25,7.5) (29,1.5)} node [midway, above, sloped] {$\varphi$ increases};
    \end{axis}
        \begin{semilogyaxis}[
    xmin=16,xmax=35,
    xshift = 8cm,
    font=\scriptsize,
    xlabel near ticks, 
    ylabel near ticks, 
    ytick={0,10,100},
    x tick label style={/pgf/number format/1000 sep=\,},
    ylabel={$\text{Var}\bracketit{z(\xi)}$}, 
    xlabel={$\mathbb{E}\bracketit{\loss}$}, 
    label style={font=\footnotesize},
    tick label style={font=\scriptsize}, 
    legend pos = south west, 
    legend style={draw=none, font=\scriptsize, 
    legend cell align={left},fill=none}, 
    width=5.6cm, height=3.5cm]
        \addplot [black, line width=0.15mm, mark=pentagon*,every mark/.append style={solid, fill=gray}, thick,mark size=1.5pt,
            mark repeat*=20, smooth] 
        table [x index = 1, y index = 2] {data_plots/summary_var_sol.dat};
        \addplot [black, line width=0.15mm, mark=*,every mark/.append style={solid, fill=gray}, thick,mark size=1.5pt,
            mark repeat*=20, smooth] 
        table [x index = 3, y index = 4] {data_plots/summary_var_sol.dat};
        \addplot [black, line width=0.15mm, mark=otimes*,every mark/.append style={solid, fill=gray}, thick,mark size=1.5pt,
            mark repeat*=20, smooth] 
        table [x index = 5, y index = 6] {data_plots/summary_var_sol.dat};
        \addplot [black, line width=0.15mm, mark=triangle*,every mark/.append style={solid, fill=gray}, thick,mark size=1.5pt,
            mark repeat*=20, smooth] 
        table [x index = 7, y index = 8] {data_plots/summary_var_sol.dat};
        \addplot [black, line width=0.15mm, mark=square*,every mark/.append style={solid, fill=gray}, thick,mark size=1.5pt,
            mark repeat*=20, smooth] 
        table [x index = 9, y index = 10] {data_plots/summary_var_sol.dat};
        \addplot[->,>=stealth, gray] coordinates
           {(23,120) (26,20)} node [midway, above, sloped] {$\varphi$ increases};
        \addplot[red,dashed] coordinates
           {(16,5) (35,5)} node [midway, above, sloped, xshift=1.0cm] {\scriptsize $\Tr[\Sigma]=4.5$};
    \end{semilogyaxis}
    \end{tikzpicture}
\caption{Trade-offs: expected value of optimality loss vs.~variance (left) and vs.~variance of optimization solution (right). The results are given for $c_{1} \!\sim\! U[1,\overline{c}]$ and averaged over 100 runs.}
% \caption{The trade-off between the variance and expected value of optimality loss (left plot) and the trade-off between the variance of optimization solution and expected optimality loss (right plot). The results are given for $c_{1}\sim U[1,\overline{c}]$ and averaged over 100 simulation runs}
\label{figure_var}
\end{figure}
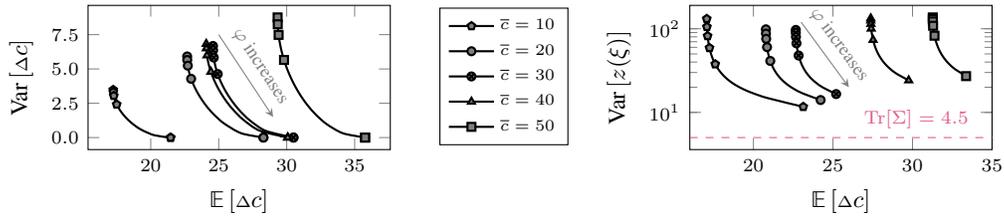

\section{Related work}
\label{sec:related_work}
There is a large body of work on differentially private algorithms for convex optimization in the context of empirical risk minimization (ERM) problems. Output perturbation algorithms \citep{chaudhuri2009privacy,rubinstein2012learning} focus on adding the noise to the optimization results. Objective perturbation algorithms \citep{chaudhuri:2011} perturb the optimization objective and perform well for smooth loss functions. Exponential sampling algorithms \citep{mcsherry2007mechanism,bassily2014private} rely on an evaluation function to select a candidate output, while achieving $(\varepsilon, 0)$-differential privacy that may be difficult to implement due to the exponential nature of evaluation function. Finally, noisy stochastic gradient descent (SGD) algorithms \citep{abadi2016deep,song2013stochastic,bassily2014private,wang2015privacy} provide a privacy-preserving version of SGD that can be combined with accountant methods to provide tight bounds. All these algorithms, however, are meant for a particular class of \emph{unconstrained} or \emph{regularized} convex optimization problems and do not focus on reporting solutions that must satisfy problem constraints. 

The contributions on differentially private \emph{constrained} convex optimization for generic decision-making problems are much more sparse. \citet{gupta2010differentially} studied differential privacy in combinatorial optimization problems and derived information-theoretic bounds on the task utility. \citet{hsu2014privately} proposed to solve linear programs privately using a differentially private variant of the multiplicative weights mechanism. \citet{han2014differentially} focused on a particular class of convex optimization problems whose objective function is piecewise affine, with the possibility of including linear inequality constraints. 
\citet{Fioretto:ArXiv20a} proposed a private data-release mechanism relying on projections to restore the feasibility of the violated constraints due to input perturbation. 
% Due to the use of input perturbation on the sensitive data, this approach modifies the original feasible space and the projections are performed in this modified feasible space.  
Finally, \citet{munozprivate} developed a differentially private algorithm for a class of linear programs that solely include the inequality constraints whose right-hand side contains sensitive data. The work relies on the input perturbation of the inequality right-hand sides to achive $(\varepsilon, \delta)-$differential privacy.

There are also differential privacy proposals for the distributed convex optimization. A privacy-preserving version of the alternating direction method of multipliers \citep{boyd2011distributed} has been studied in the context of the unconstrained ERM problem  \citep{zhang2016dynamic,ding2019optimal} and constrained energy resource allocation problem \citep{dvorkin2019differentially}. \citet{han2016differentially} proposed a private distributed projected gradient descent algorithm for constrained convex optimization problems. This collection of work, however, minimizes the privacy leakage by acting on the information exchanged by agents during the coordination process and does not provide privacy guarantees for the release of optimization solution. 

\section{Conclusion}

The paper proposed a novel framework to release privacy-preserving solutions of constrained convex optimization problems that contain complex feasibility constraints. The framework relies on a combination of differential privacy and stochastic optimization theory and provides the foundations for two algorithms answering privacy-preserving identity and sum queries over the optimization solutions. The feasibility guarantees were studied for both individual and joint constraint satisfaction and the paper examined the trade-off between the expected and the worst-case errors by controlling the variance of the solutions and the optimality loss. Finally, the proposed framework was shown to outperform standard output perturbation algorithms on several energy benchmark networks.

\small
\bibliographystyle{abbrvnat.bst}
\bibliography{lib.bib}

\normalsize
\appendix

\section{Proof of Theorem \ref{th:identity_query}}
\begin{proof}
Without loss of generality, consider that the identity query $Iz(\xi)=I\tilde{z}+IZ\xi$ requires releasing first $k$ items of $z(\xi)$, such that the diagonal matrix $I$ can be described as
\begin{align*}
I = 
\begin{bmatrix}
\text{diag}(\mathbb{1})_{k\times k} & \\
 & \text{diag}(\mathbb{0})_{n-k\times n-k}
\end{bmatrix},
\end{align*}
the perturbation vector $\xi$ as
\begin{align*}
\xi = 
\begin{bmatrix}
\xi_{1}, \dots, \xi_{k}, \mathbb{0}_{n-k}^{\top} 
\end{bmatrix}^{\top},
\end{align*}
and an arbitrary identity outcome $O$ as
\begin{align*}
O = 
\begin{bmatrix}
O_{1}, \dots, O_{k}, \mathbb{0}_{n-k \times n-k}^{\top}
\end{bmatrix}^{\top}. 
\end{align*}
Denote the optimal solution of the chance-constrained problem \eqref{cc_problem_v2}  by $\optimal{\tilde{z}}$ and $\optimal{Z}$\;. 
It needs to be shown that the ratio of probabilities 
that the algorithm returns the same outcome $O$ on two $\alpha-$indistinguishable input datasets $d$ and  $d'$ is bounded by a constant $\text{exp}\braceit{\varepsilon}$:
\begin{align*}
    \mathbb{P}\bracketit{I\braceit{\optimal{\tilde{z}}(d) + \optimal{Z}(d)\xi} = O} /\;
    \mathbb{P}\bracketit{I\braceit{\optimal{\tilde{z}}(d') + \optimal{Z}(d')\xi} = O}
    \leqslant \text{exp}\braceit{\varepsilon}.
\end{align*}
It follows that: 
\begin{align*}
    &\mathbb{P}\bracketit{I\optimal{\tilde{z}}(d) + I\optimal{Z}(d)\xi = O} /\;
    \mathbb{P}\bracketit{I\optimal{\tilde{z}}(d') + I\optimal{Z}(d')\xi = O}
    \\
    &\overset{(\text{i})}{=}
    \frac{\mathbb{P}\bracketit{\begin{bmatrix}\optimal{\tilde{z}}_{1}(d) \\ \vdots \\ \optimal{\tilde{z}}_{k}(d) \\ \mathbb{0}_{n-k}\end{bmatrix} + \begin{bmatrix}
    \xi_{1} \\ \vdots \\ \xi_{k} \\ \mathbb{0}_{n-k} 
    \end{bmatrix} = 
    \begin{bmatrix}
    O_{1} \\ \vdots \\ O_{k} \\ \mathbb{0}_{n-k}
    \end{bmatrix}}}
    {\mathbb{P}\bracketit{\begin{bmatrix}\optimal{\tilde{z}}_{1}(d') \\ \vdots \\ \optimal{\tilde{z}}_{k}(d') \\ \mathbb{0}_{n-k}\end{bmatrix} + \begin{bmatrix}
    \xi_{1} \\ \vdots \\ \xi_{k} \\ \mathbb{0}_{n-k} 
    \end{bmatrix} = 
    \begin{bmatrix}
    O_{1} \\ \vdots \\ O_{k} \\ \mathbb{0}_{n-k}
    \end{bmatrix}
    }} 
    \overset{(\text{ii})}{=} 
    \frac{
    \mathbb{P}\bracketit{
    \begin{bmatrix}
    \xi_{1} \\ \vdots \\ \xi_{k}
    \end{bmatrix}
    =
    \begin{bmatrix}
    O_{1} \\ \vdots \\ O_{k}
    \end{bmatrix}
    -
    \begin{bmatrix}\optimal{\tilde{z}}_{1}(d) \\ \vdots \\ \optimal{\tilde{z}}_{k}(d) \end{bmatrix}
    }}
    {
    \mathbb{P}\bracketit{
    \begin{bmatrix}
    \xi_{1} \\ \vdots \\ \xi_{k}
    \end{bmatrix}
    =
    \begin{bmatrix}
    O_{1} \\ \vdots \\ O_{k}
    \end{bmatrix}
    -
    \begin{bmatrix}\optimal{\tilde{z}}_{1}(d') \\ \vdots \\ \optimal{\tilde{z}}_{k}(d') \end{bmatrix}
    }}
    \\
    &\overset{(\text{iii})}{=}
    \frac{\prod_{i=1}^{k}\text{exp}\braceit{-\frac{\varepsilon\norm{O_{i}-\optimal{\tilde{z}}_{i}(d)}_{1}}{\Delta_{\alpha}}}}
    {\prod_{i=1}^{k}\text{exp}\braceit{-\frac{\varepsilon\norm{O_{i}-\optimal{\tilde{z}}_{i}(d')}_{1}}{\Delta_{\alpha}}}}
    =
    \prod_{i=1}^{k}\text{exp}\braceit{\frac{\varepsilon\norm{O_{i}-\optimal{\tilde{z}}_{i}(d')}_{1} - \varepsilon\norm{O_{i}-\optimal{\tilde{z}}_{i}(d)}_{1}}{\Delta_{\alpha}}} 
    \\
    &\overset{(\text{iv})}{\leqslant}
    \prod_{i=1}^{k}\text{exp}\braceit{\frac{\varepsilon\norm{\optimal{\tilde{z}}_{i}(d) - \optimal{\tilde{z}}_{i}(d')}_{1}}{\Delta_{\alpha}}}
    =
    \text{exp}\braceit{\frac{\varepsilon\norm{\optimal{\tilde{z}}(d) - \optimal{\tilde{z}}(d')}_{1}}{\Delta_{\alpha}}}
    \overset{(\text{v})}{\leqslant}\text{exp}\braceit{\varepsilon},
\end{align*}
where (i) is obtained from the primal feasibility condition $Z\in\mathcal{Q}(I)$, which enforces independence between the query random component and the sensitive data (see Definition \ref{def:identity_query}), (ii) comes from rearranging the terms and removing zero entries, (iii) is due to the definition of the probability density function of the Laplace distribution, (iv) follows the reverse inequality of norms, and (v) is from the definition of $\ell_{1}-$sensitivity on $\alpha-$indistinguishable input datasets.
\end{proof}

\section{Proof of Theorem \ref{th:affine_transofrmation_query}}
\begin{proof}
Without loss of generality, consider that the sum query $Sz(\xi)=S(\tilde{z} + Z\xi)$ requires releasing $p$ amount of sum statistics over non-intersecting subsets of $z(\xi)$. We thus need to show that the ratio of probabilities that the algorithm returns the same outcome $O\in\mathbb{R}^{p}$, i.e., 
\begin{align*}
    \mathbb{P}\bracketit{S(\optimal{\tilde{z}}(d)+\optimal{Z}(d))\xi = O}/\;\mathbb{P}\bracketit{S(\optimal{\tilde{z}}(d')+\optimal{Z}(d'))\xi = O}\leqslant\text{exp}(\varepsilon),
\end{align*}
is bounded by a constant $\text{exp}(\varepsilon)$, where $S\in\mathbb{R}^{p\times n}$ and $\xi\in\mathbb{R}^{p}$ as in Defintion \ref{def:sum_query}. By denoting the optimal solution of the chance-constrained problem \eqref{cc_problem_v2}  by $\optimal{\tilde{z}}$ and $\optimal{Z}\;$, this ratio writes as 
\begin{align*}
    &\frac{\mathbb{P}\bracketit{S\optimal{\tilde{z}}(d)+S\optimal{Z}(d)\xi = O}}{\mathbb{P}\bracketit{S\optimal{\tilde{z}}(d')+S\optimal{Z}(d')\xi = O}} 
    \overset{(\text{i})}{=}
    \frac{\mathbb{P}\bracketit{S\optimal{\tilde{z}}(d)+\xi = O}}{\mathbb{P}\bracketit{S\optimal{\tilde{z}}(d')+\xi = O}}
    = 
    \frac{\mathbb{P}\bracketit{\xi = O - S\optimal{\tilde{z}}(d)}}{\mathbb{P}\bracketit{\xi = O - S\optimal{\tilde{z}}(d')}}
    \\
    &\overset{(\text{ii})}{=} 
    \frac{\prod_{i=1}^{p}\text{exp}\braceit{-\frac{\varepsilon\norm{O_{i} - [S\optimal{\tilde{z}}\;(d)]_{i}}_{1}}{\Delta_{\alpha}}}}
    {\prod_{i=1}^{p}\text{exp}\braceit{-\frac{\varepsilon\norm{O_{i} - [S\optimal{\tilde{z}}\;(d')]_{i}}_{1}}{\Delta_{\alpha}}}}
    = 
    \prod_{i=1}^{p}\text{exp}\braceit{\frac{\varepsilon\norm{O_{i} - [S\optimal{\tilde{z}}\;(d')]_{i}}_{1}- \varepsilon\norm{O_{i} - [S\optimal{\tilde{z}}\;(d)]_{i}}_{1}}{\Delta_{\alpha}}} 
    \\
    &\overset{(\text{iii})}{\leqslant}
    \prod_{i=1}^{p}\text{exp}\braceit{\frac{\varepsilon\norm{[S\optimal{\tilde{z}}\;(d)]_{i}-[S\optimal{\tilde{z}}\;(d')]_{i}}_{1}}{\Delta_{\alpha}}} 
    =
    \text{exp}\braceit{\frac{\varepsilon\norm{S\optimal{\tilde{z}}\;(d)-S\optimal{\tilde{z}}\;(d')}_{1}}{\Delta_{\alpha}}} 
    \overset{(\text{iv})}{\leqslant} 
    \text{exp}\braceit{\varepsilon},
\end{align*}
where (i) follows from the primal feasibility condition $Z\in\mathcal{Q}(S)$, which requires the random component of the sum query to be independent from the data (see Definition \ref{def:sum_query}), (ii) is due to the definition of the probability density function of the Laplace distribution, (iii) follows from the reserve inequality of norms, and (iv) is from the $\ell_{1}-$sensitivity of the sum query, which is identical to the $\ell_{1}-$sensitivity of the identity query. 
\end{proof}

\section{Proof of Theorem \ref{th:gaussian_mech}}
Similarly to the proofs of Theorems 2 and 3, the random components of the identity and linear queries can be shown to be independent from a datasets $d$ and $d'$ using the query specific feasibility conditions $\mathcal{Q}$. The reminder of the proof can be obtained by following the same steps of the proof in \citep[Appendix A]{Dwork:06}, using notation $f(d)=I\optimal{\tilde{z}}(d)$ for the identity query and $f(d)=S\optimal{\tilde{z}}(d)$ for the sum query, where $f(\cdot)$ is the function of interest in \citep[Appendix A]{Dwork:06}. 

\section{Proof of Theorem \ref{th:composition}}\label{composition_algorithms}

\begin{figure}[tp]
\centering
\begin{minipage}{0.47\textwidth}
\begin{algorithm}[H]
  {\bf Input:} $d,\Delta_\alpha, \varepsilon, \eta, \mu, I$\\
  $(\optimal{\tilde{z}}\;, \optimal{Z}\;) \gets $
  Solve \eqref{cc_problem_v2} for $Z\in\mathcal{Q}(I)$
  and $\xi \sim \text{Lap}(T\Delta_\alpha/\varepsilon)$
  \\
  \For{$i=1,\ldots,T = \lceil\frac{\log(\mu)}{\log(\eta)}\rceil$}
  {
    $\hat{\xi} \gets$ Sample from $\text{Lap}
      (\nicefrac{T\Delta_\alpha}{\varepsilon})^{n}$\\
    Compute solution to \eqref{base_problem} as $\hat{z}=\optimal{\tilde{z}}\;+\optimal{Z}\hat{\xi}$\\
    \If{$A\hat{z}\leqslant b\; \lor\; i=T$}{
       \bf{Release:} $I\hat{z}$
     }      
  }
  \caption{Private identity query (PIQ)}
  \label{alg:identity2}
\end{algorithm}
\end{minipage}
\hspace{0.5cm}\begin{minipage}{0.47\textwidth}
\begin{algorithm}[H]
  {\bf Input:} $d,\Delta_\alpha, \varepsilon, \eta,  \mu, S$\\
  $(\optimal{\tilde{z}}\;, \optimal{Z}\;) \gets $
  Solve \eqref{cc_problem_v2} for $Z\in\mathcal{Q}(S)$
  and $\xi \sim \text{Lap}(T\Delta_\alpha \varepsilon)$
  \\
  \For{$i=1,\ldots,T = \lceil\frac{\log(\mu)}{\log(\eta)}\rceil$}
  {
    $\hat{\xi} \gets$ Sample from $\text{Lap}
      (\nicefrac{T\Delta_\alpha}{\varepsilon})^{p}$\\
    Compute solution to \eqref{base_problem} as $\hat{z}=\optimal{\tilde{z}}\;+\optimal{Z}\hat{\xi}$\\
    \If{$A\hat{z}\leqslant b\; \lor\; i = T$}{
       \bf{Release:} $S\hat{z}$
     }      
  }
  \caption{Private sum query (PSQ)}
  \label{alg:linear2}
\end{algorithm}
\end{minipage}
\end{figure}
%%%%%%%%%%%%
% Given feasibility requirement $\eta$, privacy parameter $\varepsilon/T$, and value $0 < \mu < 1$, the \emph{iterative} variants of Algorithms \ref{alg:identity} and \ref{alg:linear} return an $\varepsilon$-differentially private solution that is feasible with probability at least $1-\mu$ within $T = \lceil \frac{\log(\mu)}{\log(\eta)} \rceil$ steps.

The iterative versions of the Algorithms \ref{alg:identity} and \ref{alg:linear} are provided, respectively, in Algorithms \ref{alg:identity2} and \ref{alg:linear2}.

\begin{proof}
Consider the optimal solution $(\optimal{\tilde{z}}{}, \optimal{Z}\;)$ returned by the chance constraint problem (line 2) and recall that  the sampling process (lines 4--5) generates a $(1-\eta)$-feasible solution $\hat{z}$. 

The new algorithms, illustrated in Algorighms \ref{alg:identity2} and \ref{alg:linear2}, alternate this step with a constraint satisfaction test (line 6) for a maximum number of number $T$ of times with the goal of generating a solution that satisfies the problem constraints with probability at least $1-\mu$. Recall that the constraint satisfaction test, performed in line 6, can be achieved at no extra privacy loss (see Section \ref{sec:algorithm} for details).

The repetition of such process can be seen as a sequence of independent Bernulli trials, each with probability $1-\eta$ of success (i.e., $\hat{z}$ satisfies the problem constraints) and probability $\eta$ of failure (i.e., $\hat{z}$ violates the problem constraints). 
Let $\text{FAIL}$ be the discrete random variable describing the number of unsuccessful trials prior to the first success. Thus, $\text{FAIL}$ is described by a Geometric random variable with probability $(1-\eta)$.
Formally, the goal is described by the following problem:
\[
T\triangleq\arg\min_T \mathbb{P}(\text{FAIL} \geqslant T) \leqslant \mu, 
\]
requiring that the first success is seen after $T$ trials with probability no larger than $\mu$. Using a Geometric distribution of order $T$, it follows that:
\begin{align*}
\mathbb{P}(\text{FAIL} \geqslant T) &= 1 - \mathbb{P}(\text{FAIL} < T) \\
  &= 1 - (1-\eta) \sum_{i=1}^{T-1} \eta^i \\
  &= 1 - (1-\eta) \frac{1-\eta^T}{1-\eta} = \eta^T
\end{align*}
Thus, the solution to the minimizer above is for 
\(
T = \left\lceil \frac{\log(\mu)}{\log(\eta)} \right\rceil.
\)
\end{proof}

\endgroup
\end{document}